\documentclass[11pt]{article}
\include{macros.sty}

\usepackage[citestyle=alphabetic,bibstyle=alphabetic,minalphanames=5,maxalphanames=6, maxbibnames=99,url=false]{biblatex}
\title{Let's Have Both! \\Optimal List-Recoverability via Alphabet Permutation Codes}

\ifauthors{
\author[1]{Sergey Komech}
\author[1]{Jonathan Mosheiff}
\affil[1]{Department of Computer Science, Ben-Gurion University}
\date{\small{\it{To the memory of Professor Boris Markovich Gurevich, with infinite gratitude and respect}}}
}
\else{
\author{Anoynmous Submission}
\date{}
}
\fi

\begin{document}
	
	\maketitle
	
	\ifauthors	\blfootnote{JM is supported by Israel Science Foundation grant 3450/24 and an Alon Fellowship. SK is supported by European Research Council Grant No. 949707.}
	\fi
	
	\thispagestyle{empty}	
	\begin{abstract}

        We introduce \emph{alphabet-permutation (AP) codes}, a new family of error-correcting codes defined by iteratively applying random coordinate-wise permutations to a fixed initial word. A special case recovers random additive codes and random binary linear codes, where each permutation corresponds to an additive shift over a finite field.
        
        We show that when these permutations are drawn from a suitably ``mixing'' distribution, the resulting code is almost surely list-recoverable with list size proportional to the inverse of the gap to capacity. Compared to any linear code, our construction achieves exponentially smaller list sizes at the same rate. Previously, only fully random codes were known to attain such parameters, requiring exponentially many random bits and offering no structure. In contrast, AP codes are structured and require only polynomially many random bits---providing the first such construction to match the list-recovery guarantees of random codes.
	\end{abstract}


	
    \clearpage
	\setcounter{page}{1}

\section{Introduction}

In many modern coding scenarios—from communication over noisy channels to derandomization—each received symbol may be ambiguous, lying in a small list of possible values. This motivates the notion of \deffont{list-recoverability}, a fundamental generalization of list-decoding that has become central to modern code constructions and decoding algorithms~\cite{GS1998,GI2001,GR2008,GW2013,Kopparty2015,HW2018,HRW2020}. List-recoverable codes also play a key role in diverse applications across theoretical computer science, including the construction of pseudorandom objects such as extractors and condensers~\cite{GUV2009,TZ2004}, algorithmic applications~\cite{INR2010,NPR2011,GNP+2013,LNN+2016,DW2022}, and cryptographic constructions~\cite{MNP+2024}.

Concretely, let $\Sigma$ be a finite \deffont{alphabet} of size at least two, and let $\cC \subseteq \Sigma^n$ be a code of blocklength $n \in \N$. We say that $\cC$ is \deffont{$\LR{\rho}{\ell}{L}$} if it contains no \deffont{$\Cl{\rho}{\ell}$} subset of size $L+1$. Here, a set $D \subseteq \Sigma^n$ is \deffont{$\Cl{\rho}{\ell}$} if there exist sets $Z_1, \dots, Z_n \subseteq \Sigma$, each of size at most $\ell$, such that every $x \in D$ satisfies
\[
\bigl|\{\,i \in [n] : x_i \notin Z_i \}\bigr| \;\le\; \rho n \ecomma
\]
that is, $x$ disagrees with the list tuple $(Z_1,\dots,Z_n)$ in at most $\rho n$ positions. This formulation captures a worst-case setting in which both the list sets and the error locations may be adversarial.

This definition subsumes several classical notions. For example, $(0,\ell,L)$-list-recoverability corresponds to \deffont{zero-error $(\ell,L)$-list-recoverability}; $(\rho,1,L)$-list-recoverability coincides with \deffont{$(\rho,L)$-list-decodability}; and $\LRe{\rho}{1}{1}$ captures \deffont{unique-decodability up to radius $\rho$}.

The \deffont{rate} of a code $\cC \subseteq \Sigma^n$ is 
$R := \frac{\log_q |\cC|}{n}$, where $q := |\Sigma|$. The information-theoretic limit for list-recoverability is given by
\[
h^*_{q,\ell}\inparen{\rho}
:= \begin{cases}
\rho\,\log_q\!\inparen{\frac{q-\ell}{\rho}} + (1-\rho)\,\log_q\!\inparen{\frac{\ell}{1-\rho}} &\text{if }\rho \le 1-\frac \ell q \\
1 &\text{if }\rho > 1-\frac \ell q \eperiod
\end{cases}
\]
Namely, for any $q \ge 2$, $\ell \in \N$, $0 \le \rho \le 1-\frac\ell q$, and $\eps > 0$, there exist $q$-ary codes of rate at least $1 - h^*_{q,\ell}(\rho) - \eps$ that are $\LR{\rho}{\ell}{O(\ell/\eps)}$. This bound is achieved with high probability by a \deffont{plain random code (PRC)}, i.e., a uniformly random subset of $\Sigma^n$ of size $|\Sigma|^{Rn}$. Conversely, any $q$-ary code with rate exceeding $1 - h^*_{q,\ell}(\rho) + \eps$ cannot be $\LR{\rho}{\ell}{q^{o(n)}}$ \cite[Theorem 2.4.12]{Resch2020}.

A substantial line of work~\cite{GR2008,GW2013,RW2018,LP2020,GLS+2021,GST2023,Tam24} has aimed to construct explicit codes that approach this bound. A central goal is to construct codes of rate
\[
R = 1 - h_q^*\inparen{\rho,\ell} - \eps
\]
that are $\LR{\rho}{\ell}{L}$ with $L = O(\ell/\eps)$, matching the performance of PRCs. We refer to this target as the \deffont{Elias Bound for List Recovery}, in analogy with Elias’s classical list-decoding bound~\cite{Elias1957}. See also~\cite{MRSY24} for a recent discussion.

We distinguish between the \deffont{large alphabet} and \deffont{small alphabet} regimes depending on whether $q \ge \exp{\Omega(1/\eps)}$ or $q \le \exp{o(1/\eps)}$, respectively. In this work, we focus on the large alphabet regime. While PRCs meet the Elias Bound in this setting, the best known explicit constructions—such as Folded Reed–Solomon codes~\cite{Tam24}—fall exponentially short when $\ell \ge 2$.

Furthermore, even structured random codes face strong limitations in reaching the Elias Bound. It was shown in~\cite{CZ24} that any Folded or Plain Reed--Solomon code must satisfy 
$L \ge \ell^{\Omega(1/\eps)}$, yielding an exponential gap from the target. This lower bound was subsequently extended to \deffont{Random Linear Codes (RLCs)} by~\cite{LMS2024}, who also conjectured that it applies to all linear codes. The conjecture was later confirmed in~\cite{LS2025}, whose proof appears to extend even to the broader class of additive (i.e., closed under addition) codes over $\F_q$.

\paragraph{Our Contribution.} In this work, we construct the first family of codes—beyond plain random codes—that achieve the Elias Bound for list-recovery. Prior to our work, every known code family other than PRCs fell exponentially short of this benchmark when $\ell \ge 2$. Crucially, our construction requires only a \emph{polynomial} number of random bits, in contrast to PRCs, which require exponentially many. Moreover, our codes exhibit nontrivial combinatorial structure. In light of the aforementioned barriers, they are neither linear nor even additive. In this sense, we show that it is indeed possible to have both: the list-recovery performance of plain random codes, and the structure of a significantly more economical construction.

\begin{theorem}[Main Theorem]\label{thm:introMainLR}
Let $q,\ell,L,n \in \N$ with $q \ge 2$ and $\ell < \min\{q,L\}$. Suppose $0 \le \rho \le 1 - \frac{\ell}{q}$, and let 
$
\eta \ge \frac{2\,\log_q\!\bigl(2n/\ln2\bigr)}{n}
$ and $\delta > 0$. Define
\[
R :=
1
- h^*_{q,\ell}(\rho)
- \frac{\log_q \binom q\ell +\frac{1}{n}}{L+1}
- \eta
\ecomma
\]
and assume $k := R \cdot \log_2 q \cdot n$ is an integer.

Then there exists a random code ensemble $\cC \subseteq \Sigma^n$ (with $|\Sigma| = q$) of rate $R$, using only $O(nk(\ell\log q + \log \tfrac 1\delta))$ random bits, such that
\[
\PR{\cC \text{ is } \LR{\rho}{\ell}{L}}
\;\ge\;
1 -
\bigl(\sqrt 2\cdot k \cdot q^{-\tfrac{\eta n}{2}}\bigr)
-
\bigl(c\cdot \delta \cdot n\cdot k\cdot q^{4\ell}\bigr)
\ecomma
\]
for some universal constant $c > 0$.
\end{theorem}

By setting $\eta = \frac{2\log_q(nk)}{n}$ and $\delta = \frac{1}{c\,n^2\,k\,q^{4\ell}}$, we obtain a code ensemble that requires only $O(n^2 \log q \cdot (\ell \log q + \log n))$ random bits, and with high probability yields $\LR{\rho}{\ell}{L}$ codes of rate
\[
1 - h^*_{q,\ell}(\rho) - \frac{\ell}{L+1}\inparen{1 + o(1)}\eperiod
\]

\section{Alphabet-Permutation Codes and the Road to List-Recovery}

In this section, we introduce \emph{alphabet-permutation (AP) codes}, a new family of error-correcting codes defined by iteratively applying coordinate-wise permutations. AP codes strictly generalize random additive codes and binary random linear codes, while supporting a structured encoding that enables sharp list-recovery guarantees.

Beyond defining our construction, this section presents the conceptual and technical foundations for our main results. In particular, we identify a key mixing property of permutation ensembles that suffices for list-recovery and outline a potential-function argument establishing this connection. This provides a roadmap for the remainder of the paper.




\subsection{Motivation and Overview}

Our goal is to construct codes that attain the Elias Bound for list-recovery, matching the parameters of plain random codes, but using significantly less randomness and with internal structure. 

A natural starting point is the family of random additive codes, which achieve the Elias Bound for \emph{list-decoding} in the binary setting~\cite{GHK2011,LW2021}. (In this case, they coincide with random linear codes.) These codes are highly structured and require relatively few random bits to sample. However, over larger alphabets, additive shifts fail to sufficiently disperse codewords, and the performance of additive codes degrades sharply in the list-recovery setting. This motivates the search for new code ensembles that go beyond additivity while preserving structure and enabling stronger pseudorandom behavior.

We propose \emph{alphabet-permutation (AP) codes} as a natural generalization. An AP code is defined by a sequence of coordinate-wise permutations, where at the $i$-th encoding step, we either apply or skip a set of permutations depending on the $i$-th input bit. The full code is the set of all outputs produced by such selective compositions starting from a fixed initial word. When the permutations are drawn independently from a suitable distribution, the resulting code exhibits strong pseudorandom behavior.

The central technical insight is that if the permutations are drawn from a distribution that “mixes” certain families of subsets (formally, is \emph{$\cB$-mixing}), then the resulting code almost surely intersects each such set in only a small number of codewords. This property turns out to be sufficient to guarantee list-recovery with parameters matching those of random codes.

We formalize this principle in \cref{prop:mixing} and sketch its proof using a potential-function argument inspired by~\cite{GHS+2002,LW2021}. We then show how to instantiate this framework using permutations drawn from suitably independent or approximately independent distributions, yielding codes that are both structured and optimally list-recoverable.

\subsection{Definition and Encoding}\label{sec:Defs}

\begin{definition}[Alphabet-Permutation Code]
Fix integers $q,n,k \in \N$ with $k \le n \log_2 q$, and let $\Sigma := \{0, \dots, q - 1\}$. Let $\cS_\Sigma$ denote the set of all permutations on $\Sigma$, and fix a matrix $\Pi \in \cS_\Sigma^{k \times n}$. The \deffont{encoding function} associated with $\Pi$ is
\[
\enc_\Pi \colon \F_2^k \to \Sigma^n
\]
defined as follows: given $z \in \F_2^k$, define a sequence $y^0, \dots, y^k \in \Sigma^n$ by
\[
y^0 := (0, \dots, 0) \qquad\text{and}\qquad
y^i := 
\begin{cases}
y^{i-1} &\text{if } z_i = 0 \\
\bigl(\Pi_{i,1}(y^{i-1}_1), \dots, \Pi_{i,n}(y^{i-1}_n)\bigr) &\text{if } z_i = 1
\end{cases}
\]
for $i = 1, \dots, k$. Then $\enc_\Pi(z) := y^k$. See \cref{fig:encoding} for an illustration.

The image of $\enc_\Pi$, denoted $\cC_\Pi \subseteq \Sigma^n$, is the \deffont{alphabet-permutation code} associated with $\Pi$.
\end{definition}

\begin{remark}[Multiset semantics]\label{rem:multiset}
The map $\enc_\Pi$ need not be injective. Accordingly, we interpret $\cC_\Pi$ as a multiset in general, with multiplicities given by preimage size. Given a set $B \subseteq \Sigma^n$, we interpret $|B \cap \cC_\Pi|$ as $\inabs{\{z \in \F_2^k \mid \enc_\Pi(z) \in B\}}$.
\end{remark}

\paragraph{Iterative Generation.}
The encoding process can also be viewed as building up $\cC_\Pi$ in stages: starting from the all-zero vector, we repeatedly apply coordinate-wise permutations based on each row of $\Pi$. Formally, define a sequence of multisets $\cC_0, \dots, \cC_k \subseteq \Sigma^n$ by:
\[
\cC_0 := \{(0, \dots, 0)\} \quad\text{and}\quad
\cC_i := \cC_{i-1} \cup \bigl\{ (\Pi_{i,1}(x_1), \dots, \Pi_{i,n}(x_n)) \mid x \in \cC_{i-1} \bigr\}
\]
for $i = 1, \dots, k$. Then $\cC_k = \cC_\Pi$. We refer to the sequence $(\cC_i)_{i=0}^k$ as the \deffont{generating sequence} associated with $\Pi$.

\begin{definition}[Random AP Codes]
Let $D$ be a distribution over $\cS_\Sigma$. A code $\cC_\Pi$ is said to be a \deffont{$D$-random AP code} if $\Pi \in \cS_\Sigma^{k \times n}$ is a matrix with independent entries sampled from $D$. The associated generating sequence is called a \deffont{$D$-random generating sequence}.
\end{definition}

\begin{figure}
    \centering
    \begin{tikzpicture}[
        font=\small,
        box/.style={draw, minimum width=0.6cm, minimum height=0.6cm, fill=gray!15},
        perm/.style={draw, minimum width=0.6cm, minimum height=0.6cm, fill=gray!50},
        labelrow/.style={anchor=east},
        arrow/.style={->, thick, >=Stealth},
        shorten <=2pt, shorten >=2pt
    ]
    
    \def\n{6}          
    \def\xsep{0.9}     
    \def\rowgap{1.2}   
    \def\rowgapmid{1.6} 
    
    \def\yA{0}
    \node[labelrow] at (-1.3, \yA) {$y^0$};
    \foreach \i in {1,...,\n} {
      \node[box] (y0\i) at ({(\i - 1)*\xsep}, \yA) {0};
    }
    
    \def\yB{\yA - \rowgap}
    \node[labelrow] at (-1.3, \yB) {$z_1 = 1$};
    \foreach \i in {1,...,\n} {
      \node[perm] (p1\i) at ({(\i - 1)*\xsep}, \yB) {$\Pi_{1,\i}$};
      \draw[arrow] (y0\i.south) -- (p1\i.north);
    }
    
    \def\yC{\yB - \rowgap}
    \node[labelrow] at (-1.3, \yC) {$y^1$};
    \foreach \i in {1,...,\n} {
      \node[box] (y1\i) at ({(\i - 1)*\xsep}, \yC) {};
      \draw[arrow] (p1\i.south) -- (y1\i.north);
    }
    
    \def\yD{\yC - \rowgapmid}
    \node[labelrow] at (-1.3, \yC - 0.5*\rowgapmid) {$z_2 = 0$};
    \node[labelrow] at (-1.3, \yD) {$y^2$};
    \foreach \i in {1,...,\n} {
      \node[box] (y2\i) at ({(\i - 1)*\xsep}, \yD) {};
      \draw[arrow] (y1\i.south) -- (y2\i.north);
    }
    
    \def\yE{\yD - \rowgap}
    \node[labelrow] at (-1.3, \yE) {$z_3 = 1$};
    \foreach \i in {1,...,\n} {
      \node[perm] (p3\i) at ({(\i - 1)*\xsep}, \yE) {$\Pi_{3,\i}$};
      \draw[arrow] (y2\i.south) -- (p3\i.north);
    }
    
    \def\yF{\yE - \rowgap}
    \node[labelrow] at (-1.3, \yF) {$y^3 = \enc_\Pi(z)$};
    \foreach \i in {1,...,\n} {
      \node[box] (y3\i) at ({(\i - 1)*\xsep}, \yF) {};
      \draw[arrow] (p3\i.south) -- (y3\i.north);
    }
    
    \end{tikzpicture}
    \caption{Computing $\enc_\Pi(z)$ for $n=6$, $k=3$ and $z = (1,0,1)$.}
    \label{fig:encoding}
\end{figure}

\subsection{Additive Codes as a Special Case}

Alphabet-permutation codes strictly generalize a well-known family of structured codes: additive codes over fields of characteristic two. Let $q = 2^m$ and let $G \in \F_q^{k \times n}$ be a matrix. Define
\[
\cC := \{ x G \mid x \in \F_2^k \} \subseteq \F_q^n \ecomma
\]
which is an \deffont{additive code} of $\F_2$-dimension $k$, meaning it forms a subspace over the subfield $\F_2 \subseteq \F_q$. Every such code arises as an AP code by letting each permutation act as an additive shift: for each $(i,j)$, define $\Pi_{i,j}(z) := z + G_{i,j}$. Then $\cC_\Pi = \cC$, and the associated generating sequence agrees with the successive application of the rows of $G$ via coordinate-wise addition.

Moreover, if $G$ is sampled uniformly at random, the resulting code is a \deffont{random additive code (RAC)}. In this case, each $\Pi_{i,j}$ is an independent uniform additive shift, i.e., sampled uniformly from the set $\{ z \mapsto z     + a \mid a \in \F_q \}$. Thus, RACs correspond precisely to $D$-random AP codes where $D$ is the uniform distribution over additive shifts.

\paragraph{Generalizing Beyond Additivity.} This connection motivates the study of random AP codes as a natural extension of RACs and binary random linear codes. Unlike RACs, AP codes allow for arbitrary coordinate-wise permutations, thereby enabling greater structural flexibility. Importantly, AP codes retain the iterative encoding process reminiscent of additive codes, while stepping outside the confines of linearity or additivity.

\subsection{From Additive Codes to List-Recoverable AP Codes}

In the case $q=2$, an RAC over $\F_q$ is merely a \deffont{binary RLC}. Generating sequences for such codes (under a different guise) were analyzed in several works \cite{GHS+2002,LW2021,GMM2022} that studied their list-decodability. Notably, \cite{LW2021} proves that binary RLCs are almost surely list-decodable with excellent parameters, and, in particular, achieve the Elias Bound for list-decodability. 

Recall that a random additive code (RAC) can be viewed as a $D$-random AP code, where $D$ is the uniform distribution over additive shifts $x \mapsto x + a$. The list-decoding results of~\cite{GHS+2002,LW2021,GMM2022} for binary RLCs can be attributed to the strong mixing behavior inherent in these additive permutations.

To analyze the list-recoverability of general AP codes, we abstract this idea by introducing the concept of a \deffont{$\cB$-mixing} distribution over permutations. This allows us to characterize ensembles of AP codes that spread codewords sufficiently uniformly with respect to adversarially chosen lists.

In the next section, we define $\cB_{\rho,\ell}$, the family of "bad sets" corresponding to $(\rho, \ell)$-list-recovery violations, and establish conditions under which a random AP code avoids these sets with high probability. This leads to our main technical result showing that if the underlying permutation distribution is "mixing" in a suitable sense, then the resulting AP code is list-recoverable with parameters matching the Elias Bound.

\subsection{List-Recovery via Mixing Ensembles}

Let $\rho \in [0,1]$ and $\ell \in \N$. Define the family $\cB_{\rho,\ell}$ of \deffont{bad sets for list-recovery} as those subsets of $\Sigma^n$ of the form
\[
\bigl\{x \in \Sigma^n \mid \bigl|\{i \in [n] : x_i \notin Z_i\}\bigr| \le \rho n\bigr\}
\]
where $Z_1, \dots, Z_n \subseteq \Sigma$ are sets of size $\ell$. A code is $\LR{\rho}{\ell}{L}$ if and only if it intersects every set in $\cB_{\rho,\ell}$ in at most $L$ elements.

Given a distribution $D$ over $\cS_\Sigma$, define the \deffont{power ensemble} $D^n$ as the distribution over maps $\Sigma^n \to \Sigma^n$ of the form
\[
(x_1,\dots,x_n) \mapsto (\pi_1(x_1),\dots,\pi_n(x_n)) \ecomma
\]
where each $\pi_i$ is sampled independently from $D$.

\begin{example}\label{ex:Dplus}
Let $q = 2^m$ and $\Sigma = \F_q$. Let $D_+$ be the uniform distribution over the additive shifts $\{z \mapsto z + a \mid a \in \F_q\}$. Then the power ensemble $D_+^n$ consists of maps of the form $x \mapsto x + u$, where $u \in \F_q^n$ is uniformly random.
\end{example}

To reason about when a random AP-code yields good list-recovery guarantees, we introduce the following notion of a distribution that 'mixes' bad sets.

\begin{definition}[$\cB$-mixing distribution]
Say that a family $\cB$ of subsets of $\Sigma^n$ is \deffont{regular} if it is closed under some transitive action on $\Sigma^n$. (For example, the family $\cB_{\rho,\ell}$ is regular because it is closed under translations).

A distribution $\nu$ over bijections $\Sigma^n \to \Sigma^n$ is \deffont{$\cB$-mixing} if it satisfies:
\begin{enumerate}
    \item Every $f$ in the support of $\nu$ maps $\cB$ to itself; that is, for all $B \in \cB$, we have $f(B) \in \cB$.
    \item For each fixed $B \in \cB$, the image $f(B)$ under $f \sim \nu$ is distributed uniformly over $\cB$.
\end{enumerate}
\end{definition}

If the maps in a generating sequence are sampled from a $\cB$-mixing distribution, then the resulting AP code is likely to intersect each set in $\cB$ in only a small number of codewords. This principle forms the technical core of our work and underlies our main results. The following proposition formalizes this idea. We sketch its proof below—via a potential-function argument that generalizes techniques from~\cite{GHS+2002,LW2021}—and present the full proof in \cref{sec:ProofMain}.

\begin{restatable}[$\cB$-mixing implies small intersections]{proposition}{mixing}
      \label{prop:mixing}
    Fix $q, L, n \in \N$, and let $\Sigma := \{0,\dots, q-1\}$. Let $\cB$ be a regular family of subsets of $\Sigma^n$, each of cardinality at most $q^{\beta n}$ ($0 \le \beta \le 1$).
    Fix a distribution $D$ over $\cS_\Sigma$ such that $D^n$ is $\cB$-mixing. Let $k \in \N$ such that the rate $R := \frac{k}{n \log_2 q}$ satisfies
    \[
    R = 1 - \beta - \frac{\log_q |\cB| + 1}{n(L+1)} - \eta
    \]
    for some \begin{equation}\label{eq:etaLowerBound}
        \eta \ge \frac{2\log_q(2\cdot n/\ln2)}n\eperiod
    \end{equation} Let $\Pi \in \cS_\Sigma^{k \times n}$ be a matrix with independent entries sampled from $D$.
    Then,
    $$\PR{\exists B\in \cB \text{ such that }\inabs{\cC_\Pi \cap B} > L} \le \sqrt 2\cdot k\cdot q^{-\frac{\eta n} 2}\eperiod$$
\end{restatable}

\begin{proof}[Proof sketch for {\cref{prop:mixing}}]
Let $\cC_0,\dots,\cC_k$ be a $D$-generating sequence in $\Sigma^n$. Write $q = |\Sigma|$ and $\cB = \cB_{\rho,\ell}$. For each $0 \le i \le k$, define
\[
K_i := \Eover{B \sim \uniform(\cB)}{q^{\inabs{\cC_i \cap B} \cdot \alpha\cdot n}}
\]
for some fixed $\alpha > 0$. To prove that $\cC_k$ is $\LR{\rho}{\ell}{L}$, it suffices to show the following:
\begin{enumerate}
    \item $K_0$ is small (since $\cC_0$ is a singleton).
    \item With high probability, the sequence $K_0, \dots, K_k$ grows slowly.
    \item If $K_k$ is small, then $\cC_k$ avoids large intersections with any $B \in \cB$; in particular, $\cC_k$ is $\LR{\rho}{\ell}{L}$.
\end{enumerate}

The first and third items are straightforward. The third, in particular, follows since if there exists $B' \in \cB$ with $|\cC_k \cap B'| > L$, then
\[
K_k = \Eover{B \sim \uniform(\cB)}{q^{|\cC_k \cap B| \cdot \alpha\cdot n}} \ge \frac{q^{|\cC_k\cap B'|\cdot \alpha\cdot n}}{|\cB|}\ge \frac{q^{(L+1)\cdot \alpha\cdot n}}{|\cB|}.
\]

The main challenge lies in bounding the growth of $K_i$. We show that
\[
\Eover{\cC_i}{K_i \mid K_{i-1}} \le K_{i-1}^2.
\]
This recurrence implies that $K_i$ grows roughly quadratically in expectation. A first-moment bound and union bound over $i$ show that $K_k$ remains small with high probability.

To prove the recurrence:
\begin{align*}
\Eover{\cC_i}{K_i\mid \cC_{i-1}} 
&= \Eover{f\sim D^n}{\Eover{B\in \cB}{q^{\inabs{(\cC_{i-1} \cup f(\cC_{i-1}))\cap B}\cdot \alpha}}} \\
&\le \Eover{f\sim D^n}{\Eover{B\in \cB}{q^{\inabs{\cC_{i-1} \cap B}\cdot \alpha + \inabs{f(\cC_{i-1})\cap B}\cdot \alpha}}} \\
&= \Eover{B\in \cB}{\Eover{f\sim D^n}{q^{\inabs{\cC_{i-1} \cap B}\cdot \alpha + \inabs{f(\cC_{i-1})\cap B}\cdot \alpha}}}
\\
&= \Eover{B\in \cB}{q^{\inabs{\cC_{i-1} \cap B}\cdot \alpha}\cdot \Eover{f\sim D^n}{q^{ \inabs{f(\cC_{i-1})\cap B}\cdot \alpha}}}
\\
&=\Eover{B\in \cB}{q^{\inabs{\cC_{i-1} \cap B}\cdot \alpha}\cdot \Eover{f\sim D^n}{q^{ \inabs{f(\cC_{i-1})\cap f(B')}\cdot \alpha}}}
&&\text{(set $B' := f^{-1}(B)$)} \\
&=\Eover{B\in \cB}{q^{\inabs{\cC_{i-1} \cap B}\cdot \alpha}\cdot \Eover{f\sim D^n}{q^{ \inabs{\cC_{i-1}\cap B'}\cdot \alpha}}}
&&\text{($f$ bijective)} \\
&=\Eover{B\in \cB}{q^{\inabs{\cC_{i-1} \cap B}\cdot \alpha}\cdot \Eover{B'\in \cB}{q^{ \inabs{\cC_{i-1}\cap B'}\cdot \alpha}}}
&&\text{($f^{-1}(B)$ uniform in $\cB$)} \\
&= \Eover{B\in \cB}{q^{\inabs{\cC_{i-1} \cap B}\cdot \alpha}}\cdot \Eover{B'\in \cB}{q^{\inabs{\cC_{i-1} \cap B'}\cdot \alpha}}
&&\text{(independence)} \\
&= K_{i-1} \cdot K_{i-1} = K_{i-1}^2\eperiod
\end{align*}
\end{proof}

Taking $\cB = \cB_{\rho,\ell}$, \cref{prop:mixing} yields immediate implications for list-recovery.

\begin{corollary}[$\cB_{\rho,\ell}$-mixing implies list-recoverability]
\label{cor:mixingToLR}
    Fix $\ell, q,L, n\in \N$ and $\rho \ge 0$ such that $1\le \ell < q$ and $\rho < 1-\frac \ell q$. Let $\Sigma = \{0,\dots, q-1\}$ and fix a distribution $D$ over $\cS_\Sigma$ such that $D^n$ is $\cB_{\rho,\ell}$-mixing. Let $k\in \N$ such that $$R:= \frac{k}{n\cdot \log_2 q} = 1 - h^*_{q,\ell}\!\inparen{\rho} - \frac {\log_q\binom q \ell+\frac 1n} {L+1} - \eta$$ for some $$
        \eta \ge \frac{2\log_q(2\cdot n/\ln2)}n\eperiod
    $$ Let $\Pi \in S_\Sigma^{k\times n}$ be a random matrix whose entries are sampled independently at random from $D$.    
    Then,
    $$\PR{\cC_\Pi \textrm{ is }\LR \rho \ell L} \ge 1 - \sqrt 2\cdot k\cdot q^{-\frac{\eta n} 2}\eperiod$$
\end{corollary}
\begin{proof}[Proof of \cref{cor:mixingToLR} given \cref{prop:mixing}]
    Let $\cB = \cB_{\rho,\ell}$. By a standard estimation \cite[Prop.\ 2.4.11]{Resch2020}, every $B\in \cB$ is of size at most $q^{n\cdot h_{q,\ell}^*(\rho)}$. Moreover, since each $B \in \cB_{\rho,\ell}$ is determined by a tuple $(Z_1, \dots, Z_n)$ with $Z_i \in \binom{\Sigma}{\ell}$, we have $|\cB| \le \binom{q}{\ell}^n$.

     Recall that $\cC_\Pi$ is $\LR \rho \ell L$ if and only if $|\cC\cap B| \le L$ for every $B\in \cB$. The corollary now follows immediately from \cref{prop:mixing}.
\end{proof}

A more specialized corollary concerns list-decoding of RACs over fields of characteristic $2$.

\begin{corollary}\label{cor:linear}
Fix $t, L, n \in \N$. Set $q := 2^t$ and let $0 \le \rho \le 1 - \tfrac{1}{q}$. Let $k \in \N$ satisfy
$$
R := \frac{k}{n\,t} 
= 1 - h_q(\rho) - \frac{1 + \tfrac{1}{n}}{L+1} - \eta\eperiod
$$
for some 
$$
\eta \;\ge\; \frac{2\,\log_q\!\bigl(2\,n/\ln 2\bigr)}{n}\eperiod
$$
Let $\cC$ be a random $\F_2$-linear code in $\F_q^n$ of rate $R$. Then,
$$
\PR{\cC \text{ is } \LD{\rho}{L}}
\;\ge\;
1 - \sqrt{2}\,k \,q^{-\tfrac{\eta n}{2}}\eperiod
$$
\end{corollary}
\begin{proof}[Proof of \cref{cor:linear} given \cref{prop:mixing}]
We instantiate \cref{cor:mixingToLR} by taking $\Sigma = \F_q$ and setting $\ell = 1$, noting that list-decodability corresponds to list-recoverability with $\ell = 1$. We let $D = D_+$, the uniform distribution over additive shifts $z \mapsto z + a$ for $a \in \F_q$.

As noted in $\cref{ex:Dplus}$, when $\Sigma = \F_q$ with $q = 2^m$, the ensemble $D_+^n$ consists of maps of the form $x \mapsto x + u$, where $u \in \F_q^n$ is chosen uniformly at random. It is easy to verify that this ensemble is $\cB_{\rho,1}$-mixing for all $\rho \in [0,1]$. 
\end{proof}

When $t=1$, $\cC$ is simply a random linear code in $\F_2^n$, thus recovering (in spirit) the main result of~\cite{LW2021}. We note that~\cite{LW2021} achieves a somewhat smaller list size by exploiting the linearity\footnote{In our framework, this corresponds to the fact that for linear (or additive) codes, the coordinate-wise permutations used in the construction commute, a property not shared by general AP codes.}
of $\cC$, a method not applicable in the more general setting of alphabet-permutation codes.

\subsection{Achieving Mixing via Independent Permutations}

While $D_+^n$ works perfectly for $\ell = 1$, the situation is more subtle when $\ell > 1$.
Unfortunately, $D_+^n$ is generally not $\cB_{\rho,\ell}$-mixing for $\ell \ge 2$. For example, let $q = 2^m$ with $m > 1$, and fix $\rho = 0$ and $\ell = 2$. Consider the combinatorial rectangles
\[
R := \{a,b\}^n \qquad\text{and}\qquad R' := \{a,c\} \times \{a,b\}^{n-1}
\]
for some $a,b,c \in \F_q$ such that $b - a \ne \pm(c - a)$. Both $R$ and $R'$ lie in $\cB_{0,2}$, yet there is no additive shift $x \mapsto x + u$ with $u \in \F_q^n$ that maps $R$ to $R'$. This shows that $D_+^n$ does not mix $\cB_{0,2}$ uniformly, and hence fails to be $\cB_{\rho,\ell}$-mixing in general when $\ell > 1$.

This motivates the search for distributions $D$ whose power ensemble $D^n$ is genuinely $\cB_{\rho,\ell}$-mixing even for $\ell > 1$, enabling list-recovery beyond what additive codes support. In the next section, we identify such ensembles by leveraging the classical notion of $\ell$-wise independence among permutations.

\begin{definition}[$m$-wise independence]\label{def:independence}
  Let $\Sigma$ be a set with $\inabs{\Sigma} = q$, and let $1 \le m \le q$. Denote by $\Sigma_m$ the set of all $m$-tuples of distinct elements in $\Sigma$. Let $D$ be a distribution over $\cS_\Sigma$. 
  We say that $D$ is \deffont{$m$-wise independent} if, for every $(x_1,\dots,x_m) \in \Sigma_m$, when $\pi$ is drawn from $D$, the tuple $(\pi(x_1),\dots,\pi(x_m))$ is uniformly distributed over $\Sigma_m$. 
\end{definition}

\begin{remark}
    $m$-wise independence implies $m'$-wise independence for all $1 \le m' \le m$.
\end{remark}

\begin{example}\label{ex:independence}
    The uniform distribution over $\cS_\Sigma$ is clearly $|\Sigma|$-wise independent. For $\Sigma = \F_q$ with $q = 2^m$ and $m > 1$, the additive shift distribution $D_+$ is $1$-wise independent but not $2$-wise independent.
\end{example}

We next show that $\ell$-wise independence suffices to ensure $\cB_{\rho,\ell}$-mixing. This connects our framework to a well-studied pseudorandomness notion and allows efficient instantiations of suitable ensembles.

\begin{lemma}\label{lem:IndpendenceYieldsMixing}
If $D$ is an $\ell$-wise independent distribution over $\cS_\Sigma$, then $D^n$ is $\cB_{\rho,\ell}$-mixing for all $\rho \in [0,1]$.
\end{lemma}

\begin{proof}
Let $B \in \cB_{\rho,\ell}$ be defined by sets $Z_1,\dots,Z_n \subseteq \Sigma$ with $\inabs{Z_i} = \ell$. Let $f = (\pi_1,\dots,\pi_n) \sim D^n$. Then,
\[
f(B) = \{x \in \Sigma^n : \inabs{\{i \in [n] : x_i \notin \pi_i(Z_i)\}} \le \rho n\} \eperiod
\]
Since each $\pi_i(Z_i)$ is uniformly distributed over $\binom{\Sigma}{\ell}$ and the sets are independent, the image $f(B)$ is uniformly distributed in $\cB_{\rho,\ell}$.
\end{proof}

\cref{cor:mixingToLR,lem:IndpendenceYieldsMixing} immediately imply that a $(D,k,n)$-random AP code achieves the list-recovery Elias bound with high probability, provided that $D$ is $\ell$-wise independent. 
\begin{theorem}\label{thm:LRFullIndependence}
    Fix $\ell, q,L, n\in \N$ and $\rho \ge 0$ such that $1\le \ell < q$ and $\rho < 1-\frac \ell q$. Let $\Sigma = \{0,\dots, q-1\}$ and fix an $\ell$-wise independent distribution $D$ over $\cS_\Sigma$. Let $k\in \N$ such that $$R:= \frac{k}{n\cdot \log_2 q} = 1 - h^*_{q,\ell}\!\inparen{\rho} - \frac {\log_q\binom q \ell+\frac 1n} {L+1} - \eta$$ for some 
    $$
        \eta \ge \frac{2\log_q(2\cdot n/\ln2)}n\eperiod
    $$ Let $\Pi \in S_\Sigma^{k\times n}$ be a random matrix whose entries are sampled independently at random from $D$.    
    Then,
    $$\PR{\cC_\Pi \textrm{ is }\LR \rho \ell L} \ge 1 - \sqrt 2\cdot k\cdot q^{-\frac{\eta n} 2}\eperiod$$
\end{theorem}

\subsection{Proof of \cref{thm:introMainLR}: Partial Derandomization via Near-Independence}

We now prove \cref{thm:introMainLR}, as stated in the introduction. Theorem~\ref{thm:LRFullIndependence} shows that $\ell$-wise independent permutations suffice to construct list-recoverable codes matching the Elias bound. However, sampling each entry of $\Pi \in \cS_\Sigma^{k \times n}$ from the uniform distribution over $\cS_\Sigma$ requires specifying a random permutation over an alphabet of size $q$, costing $O(q \log q)$ bits per entry. This yields a total randomness cost of $O(nk \cdot q \log q)$ bits.

Our goal is to reduce this to $O(nk(\ell \log q + \log \tfrac{1}{\delta}))$ bits by using distributions over $\cS_\Sigma$ that are only approximately $\ell$-wise independent. To this end, we rely on the standard notion of \emph{approximate} or \emph{near} $m$-wise independence—distributions that are close (in total variation distance) to truly $m$-wise independent ones.

\begin{definition}[Total variation distance]
Let $D$ and $D'$ be distributions over a finite set $\Omega$. The \deffont{total variation distance} between $D$ and $D'$ is
\[
\|D - D'\| := \frac{1}{2} \sum_{\omega \in \Omega} |D(\omega) - D'(\omega)|\eperiod
\]
We say that $D$ and $D'$ are \deffont{$\delta$-close} if $\|D - D'\| \le \delta$.
\end{definition}

\begin{definition}[$m$-wise $\delta$-independence of permutations]
Let $\Sigma$ be a finite set with $|\Sigma| = q$, and let $1 \le m \le q$. Let $D$ be a distribution over $\cS_\Sigma$, and fix $\delta \ge 0$. We say that $D$ is \deffont{$m$-wise $\delta$-independent} if, for every $(x_1, \dots, x_m) \in \Sigma_m$ (i.e., all entries distinct), the joint distribution of $(\pi(x_1), \dots, \pi(x_m))$ for $\pi \sim D$ is $\delta$-close (in total variation distance) to uniform over $\Sigma_m$.

A finite family $T \subseteq \cS_\Sigma$ is $m$-wise $\delta$-independent if the uniform distribution over $T$ is.
\end{definition}

\begin{remark}
$m$-wise $0$-independence coincides with plain $m$-wise independence.
\end{remark}

We now combine two tools: one giving small $m$-wise $\delta$-independent families, and another converting near-independence into true independence with bounded total variation error.

\begin{theorem}[Existence of small $\delta$-independent families {\cite[Thm.\ 5.9]{KNR2009}}]\label{thm:KNR}
Let $\Sigma$ be a finite set with $q := |\Sigma|$, and fix $1 \le m \le q$ and $\delta > 0$. There exists an $m$-wise $\delta$-independent family $T \subseteq \cS_\Sigma$ with
\[
\log_2 |T| \le O\left(m \log q + \log \tfrac{1}{\delta}\right).
\]
Moreover, each $\pi \in T$ can be evaluated in time $\polylog(|T|)$.
\end{theorem}

\begin{theorem}[Reduction to full independence {\cite{AL2013}}]\label{thm:AL}
Let $D$ be an $m$-wise $\delta$-independent distribution over $\cS_\Sigma$. Then there exists an $m$-wise independent distribution $D'$ over $\cS_\Sigma$ such that
\[
\|D - D'\| \le O(\delta \cdot q^{4m})\eperiod
\]
\end{theorem}

We now construct list-recoverable codes using a near-independent distribution over permutations. To analyze them, we couple this distribution to a truly 
$\ell$-wise independent one and apply \cref{thm:LRFullIndependence} to the latter, transferring the guarantee via a total variation bound.

\begin{proof}[Proof of \cref{thm:introMainLR}]
Let $T \subseteq \cS_\Sigma$ be an $\ell$-wise $\delta$-independent family from \cref{thm:KNR}. Let $\Pi \in \cS_\Sigma^{k \times n}$ be a matrix with independent entries sampled uniformly from $T$.

The total number of random bits required to sample $\Pi$ is at most
\[
n \cdot k \cdot \log_2 |T| = O\left(n^2 \log q \cdot \left(\ell \log q + \log \tfrac{1}{\delta}\right)\right).
\]

Let $D$ denote the uniform distribution over $T$, and let $D'$ be an $\ell$-wise independent distribution guaranteed by \cref{thm:AL}, satisfying
\[
\|D - D'\| \le O(\delta \cdot q^{4\ell})\eperiod
\]

Let $\Pi'$ be a matrix with i.i.d.\ entries drawn from $D'$, and note that by \cref{thm:LRFullIndependence},
\[
\Pr[\cC_{\Pi'} \text{ is } \LR{\rho}{\ell}{L}] \ge 1 - \sqrt{2} \cdot k \cdot q^{-\eta n/2}\eperiod
\]

Since each entry of $\Pi$ differs in total variation distance from the corresponding entry of $\Pi'$ by at most $O(\delta \cdot q^{4\ell})$, the total variation distance between the joint distributions of $\Pi$ and $\Pi'$ is $O(nk \cdot \delta q^{4\ell})$. Hence,
\begin{align*}
\Pr[\cC_\Pi \text{ is } \LR{\rho}{\ell}{L}]
\;&\ge\;
\Pr[\cC_{\Pi'} \text{ is } \LR{\rho}{\ell}{L}] - O(nk \cdot \delta \cdot q^{4\ell})
\\
&\ge 1 - \sqrt{2} \cdot k \cdot q^{-\eta n/2} - O(nk \cdot \delta \cdot q^{4\ell})\eperiod
\end{align*}
\end{proof}

\section{Mixing Implies $\cB$-Avoidance---Proof of \cref{prop:mixing}}\label{sec:ProofMain}
We shall now restate and prove \cref{prop:mixing}---the technical core of this work.

\mixing*
\begin{proof}

We define some notation, inspired by \cite{GHS+2002,LW2021}. 
For $B\in \cB$ and a code $\cC\subseteq \Sigma^n$, let 
$$
A_\cC(B) = q^{|B\cap \cC|\cdot \alpha\cdot n}
$$
where 
$$
\alpha = \frac {\log_q |\cB|+1} {(L+1)\cdot n}\ecomma
$$
and define the potential function
$$
K_\cC = \frac{1}{|\cB|}\sum_{B\in \cB}A_\cC(B)\eperiod
$$

\cref{prop:mixing} is now an immediate consequence of the two following lemmas.

\begin{lemma} \label{lem:SmallPotentialImpliesLR}
    Suppose that a code $\cC\subseteq \Sigma^n$ satisfies $K_\cC < 2$. Then $|\cC\cap B|\le L$ for all $B\in \cB$.
\end{lemma}

\begin{lemma} \label{lem:smallPotentialLikely}
    $$\PR{K_{\cC_\Pi} < 2} \ge 1-\sqrt2\cdot k \cdot q^{-\frac{\eta n} 2}$$
\end{lemma}

The rest of this section is devoted to proving both lemmas.





\begin{proof}[Proof of \cref{lem:SmallPotentialImpliesLR}]
    \sloppy
    We prove the statement in its contrapositive form. Suppose that $\inabs{\cC\cap B'} > L$ for some $B'\in \cB$. Then,

    \[K_\cC = \Eover{B \sim \uniform(\cB)}{q^{|\cC \cap B| \cdot \alpha\cdot n}} \ge \frac{q^{|\cC\cap B'|\cdot \alpha\cdot n}}{|\cB|} \ge \frac{q^{(L+1)\cdot \alpha\cdot n}}{|\cB|} = q \ge 2.
    \]
\end{proof}

\begin{proof}[Proof of \cref{lem:smallPotentialLikely}]
     Let $\cC_0,\dots,\cC_k$ denote the $\Pi$-generating sequence (see \cref{sec:Defs}). Namely,
     $\cC_0 = \{(0,\dots,0)\}$
     and 
     $$\cC_i = \cC_{i-1} \cup \tau_i(\cC_{i-1})\ecomma$$
     for $1\le i\le k$, where $\tau_1,\dots,\tau_k$ are sampled independently from $D^n$.

    Consider the sequence of real numbers $\lambda_0,\dots, \lambda_k$ defined by 
    \begin{align*}
        \lambda_0 &= q^{n\cdot \inparen{\alpha + \beta-1}} \\
        \lambda_i &= 2\lambda_{i-1} + \lambda_{i-1}^{1.5} & \text{for }1\le i\le k\eperiod
    \end{align*}
    We claim that, with high probability, \begin{equation}\label{eq:lambdaBound}
        K_{\cC_i} \le 1 + \lambda_i\quad\quad \text{for all }0\le i\le k\eperiod
    \end{equation} 
    The following claim shows that \cref{eq:lambdaBound} holds deterministically for $i=0$.
    \begin{claim}\label{claim:FirstElement}
        $$K_{\cC_0} \le 1 + q^{n\cdot \inparen{\alpha + \beta-1}}\eperiod$$
    \end{claim}
    \begin{proof}
        Note that $\cC_0$ is merely the code $\{0\}$. Recall that $\cB$ is regular, so it is closed under the action of some group $G$ that acts transitively on $\Sigma^n$. Let $\cO\subseteq \cB$ be an orbit of $\cB$ under this action. Fix some $B_0\in \cO$. Then,
        $$\sum_{B\in \cO} \ind{\vec 0\in B} = \frac{|\cO|}{|G|}\cdot \sum_{g\in G} \ind{\vec 0\in gB_0} = \frac{|\cO|}{|G|}\cdot \sum_{g\in G} \ind{g^{-1}\vec 0\in B_0} = \frac{|\cO|}{q^n}\cdot |B_0| \le \frac{|\cO|}{q^n}\cdot q^{\beta\cdot n} = |\cO|\cdot q^{(\beta-1)n}\ecomma$$
        where $\ind E$ is an indicator variable for the event $E$ and $\vec 0\in \Sigma^n$ is the all-zeros vector. Here, the penultimate step is by transitivity of $G$ on $\Sigma^n$, and the last step is by hypothesis.
        
        Thus,
        $$
        \PROver{B\sim \uniform(\cB)}{\vec 0\in B} = \frac{1}{|\cB|}\cdot \sum_{\cO} \sum_{B\in \cO} \ind{\vec 0\in B} \le \frac{q^{(\beta-1)n}}{|\cB|}\cdot \sum_{\cO} |\cO| = q^{(\beta-1)n}\ecomma
        $$
        where $\cO$ runs over all orbits of $\cB$ with regard to the action of $G$. Therefore,
        \begin{align*}
        K_{\cC_0} = \Eover{B\sim \uniform(\cB)} {q^{\alpha\cdot n\cdot |\cC_0\cap B|}} &= \PROver{B\sim \uniform(\cB)}{\vec0\notin B}\cdot 1 + \PROver{B\sim \uniform(\cB)}{\vec 0\in B}\cdot q^{\alpha \cdot n}\\ &\le 1+ \PROver{B\sim \uniform(\cB)}{\vec 0\in B}\cdot q^{\alpha \cdot n} \le 1 + q^{(\alpha + \beta-1)n}\eperiod
        \end{align*}    
    \end{proof}
    We now prove that \cref{eq:lambdaBound} holds with high probability for any $1\le i\le k$, provided that it holds for $i-1$.
    \begin{claim}\label{claim:Step}
        For all $1\le i\le k$, it holds that
        $$\PR{K_{\cC_i} > 1+\lambda_i\;\mid\;K_{\cC_{i-1}} \le 1 + \lambda_{i-1}} \le \lambda_{i-1}^{1/2}\eperiod$$
    \end{claim}

    \begin{proof}
   Fixing $B\in \cB$ and conditioning on $C_{i-1}$, we have
    \begin{align*}
        \Eover{\tau_i}{A_{\cC_i}(B)} &= \Eover{\tau_i}{q^{\alpha\cdot n\cdot|B\cap \cC_i|}} \\ &= \Eover{\tau_i}{q^{\alpha\cdot n\cdot |B\cap \cC_{i-1}|+\alpha\cdot n\cdot \inabs{B\cap \tau_i(\cC_{i-1})}}} \\ 
        &= A_{\cC_{i-1}}(B)\cdot\Eover{\tau_i}{ q^{\alpha\cdot n\cdot \inabs{B\cap \tau_i(\cC_{i-1})}}}  \\ 
        &= A_{\cC_{i-1}}(B)\cdot\Eover{\tau_i}{ q^{\alpha\cdot n\cdot \inabs{\tau_i^{-1}(B)\cap \cC_{i-1}}}}  &\textrm{ since $\tau_i$ is bijective on $\Sigma^n$}\\ 
        &= A_{\cC_{i-1}}(B)\cdot\Eover{B'\sim \uniform(\cB)}{ q^{\alpha\cdot n\cdot \inabs{B'\cap \cC_{i-1}}}} &\textrm{ taking $B' = \tau_i^{-1}(B)$}\\
        &= A_{\cC_{i-1}}(B)\cdot\Eover{B' \sim \uniform(\cB)}{A_{\cC_{i-1}}(B') }\eperiod
    \end{align*}    
    In the penultimate transition we used the fact that $D^n$ is $\cB$-mixing and that $\tau_i$ is independent from $\cC_{i-1}$. Now, still conditioning on $\cC_{i-1}$, we have
    \begin{align*}
    \Eover{\tau_i}{K_{\cC_i}} &= \Eover{\tau_i}{\Eover{B\sim \uniform(\cB)}{{A_{\cC_i}(B)}}} =  \Eover{B\sim \uniform(\cB)}{\Eover{\tau_i}{{A_{\cC_i}(B)}}}
    =
    \Eover{B\sim \uniform(\cB)}{A_{\cC_{i-1}}(B)\cdot\Eover{B'\sim \uniform(\cB)}{A_{\cC_{i-1}}(B')}} \\
    &= \Eover{B\sim U(\cB)}{A_{\cC_{i-1}}(B)} \cdot \Eover{B'\sim U(\cB)}{A_{\cC_{i-1}}(B')} = \inparen{\Eover{B\sim U(\cB)}{A_{\cC_{i-1}}(B)}}^2 = K_{\cC_{i-1}}^2\eperiod
    \end{align*}

    Write $K_{\cC_{i-1}} = 1 + \beta$. By assumption, $0\le \beta \le \lambda_{i-1}$. Note that $K_{\cC_{i}} \ge 1 + 2\beta$ deterministically. Indeed,
    \begin{align*}
    0 &\le \Eover{B\sim \uniform(\cB)}{(A_{\cC_{i-1}}(B)-1)\cdot (A_{\tau_i(\cC_{i-1})}(B)-1)} \\&= \Eover{B\sim \uniform(\cB)}{A_{\cC_{i-1}\cup \tau_i(\cC_{i-1})}(B) - A_{\cC_{i-1}}(B) - A_{\tau_i(\cC_{i-1})}(B)} + 1  \\&= \Eover{B\sim \uniform(\cB)}{A_{\cC_i}(B) - A_{\cC_{i-1}}(B) - A_{\tau_i(\cC_{i-1})}(B)} + 1 \\
    &= K_{\cC_i} - K_{\cC_{i-1}} - K_{\tau_i(\cC_{i-1})} + 1 \\
    &= K_{\cC_i} - 2(1+\beta) + 1 = K_{\cC_i} - (1+2\beta)\eperiod
    \end{align*}
    Here, the  inequality is due to $A_{\cC}(B) \ge 1$ for all $\cC$ and $B$. The first equality is since, for any two codes $\cC$ and $\cC'$, there holds
    $$A_{\cC\cup \cC'}(B) = q^{\alpha\cdot n\cdot\inabs{(\cC\cup\cC')\cap B}} = q^{\alpha\cdot n\cdot\inparen{\inabs{\cC\cap B}+\inabs{\cC'\cap B}}} = A_{\cC}(B) \cdot A_{\cC'}(B)$$
    (recall that the codes are multisets).    
    The penultimate equality is since
    $$K_{\tau_i(\cC_{i-1})} = \Eover{B\sim \uniform(\cB)}{q^{\alpha\cdot n\cdot\inabs{\tau_i(\cC_{i-1})\cap B}}} = \Eover{B\sim \uniform(\cB)}{q^{\alpha\cdot n\cdot\inabs{\cC_{i-1}\cap \tau_i^{-1}(B)}}} = \Eover{B'\sim \uniform(\cB)}{q^{\alpha\cdot n\cdot\inabs{\cC_{i-1}\cap B'}}} = K_{\cC_{i-1}}\ecomma$$
    where we took $B' = \tau_i^{-1}(B)$ and used the fact that $\tau^{-1}$ acts bijectively on $\cB$.
    
    Markov's inequality thus yields
    \begin{align*}
    \PROver{\tau_i}{K_{\cC_i} > 1+\lambda_i} &= \PROver{\tau_i}{K_{\cC_i} - (1+2\beta) > \lambda_i-2\beta} \le \frac{\Eover{\tau_i}{K_{\cC_i}}-(1+2\beta)}{\lambda_i-2\beta} \\&\le \frac{K_{\cC_{i-1}}^2-(1+2\beta)}{\lambda_i-2\beta} = \frac{(1+\beta)^2-(1+2\beta)}{\lambda_i-2\beta} \\
    &= \frac{\beta^2}{\lambda_i - 2\beta}  = \frac{\beta^2}{2\lambda_{i-1}+\lambda_{i-1}^{1.5}-2\beta} \le \frac{\beta^2}{\lambda_{i-1}^{1.5}} \le \frac{\lambda_{i-1}^2}{\lambda_{i-1}^{1.5}} = \lambda_{i-1}^{1/2}\eperiod
    \end{align*}   
    \end{proof}

    By \cref{claim:FirstElement,claim:Step},
    \begin{align*}
        \PR{K_{\cC_k} > 1+\lambda_k} \le \sum_{i=1}^k \PR{K_{\cC_i} > 1+\lambda _i \;\mid\; K_{\cC_{i-1}}\le 1+\lambda_{i-1}} \le \sum_{i=1}^k \lambda_{i-1}^{1/2} \le k\cdot \lambda_k^{1/2}\eperiod \numberthis \label{eq:K_C}
    \end{align*}
    To conclude the lemma we need the inequality
    \begin{equation}\label{eq:deltak}
        \lambda_k \le 2^{k+1}\cdot \lambda_0 = 2\cdot q^{n\cdot\inparen{\alpha + \beta-1+R}} = 2 \cdot q^{-\eta n}\le 1\ \eperiod
    \end{equation}
    Indeed, assuming that \cref{eq:deltak} holds, \cref{eq:K_C,eq:etaLowerBound} yield
    $$\PR{K_{\cC_\Pi} \ge 2} = \PR{K_{\cC_k} \ge 2} \le \PR{K_{\cC_\Pi} > 1+\lambda_k} \le k\cdot \lambda_k^{1/2} \le \sqrt2\cdot k\cdot q^{-\frac {\eta n}2}\eperiod$$
    
    We prove \cref{eq:deltak} as a special case of the more general claim that $\lambda_i \le 2^{i+1}\cdot \lambda_0$ for all $0\le i\le k$. We prove the latter by induction on $i$. The base case $i=0$ is immediate. For $1\le i\le k$, we have
    \begin{align*}\lambda_i &= 2\lambda_{i-1} + \lambda_{i-1}^{1.5} = 2\lambda_{i-1}\inparen{1 + \frac{\lambda_{i-1}^{1/2}}{2}} = 2^i\cdot \lambda_0\cdot\prod_{j=0}^{i-1}\inparen{1+\frac{\lambda_j^{1/2}}2} \leq 2^i\cdot \lambda_0\cdot \exp{\sum_{j=0}^{i-1} \frac{\lambda_j^{1/2}}2} \\
    &\le 2^i\cdot \lambda_0\cdot \exp{i\cdot 2^{\frac {i}{2}+1}\cdot \lambda_0^{\frac 12}} \le 2^i\cdot \lambda_0\cdot \exp{k\cdot 2\cdot q^{\frac n2\cdot \inparen{\alpha+\beta-1+R}}}\le 2^i\cdot \lambda_0\cdot \exp{k\cdot 2\cdot q^{-\frac{\eta n}2}}\\&\le 2^{i+1}\cdot \lambda_0\eperiod
    \end{align*}
        Here, the second inequality is by the induction hypothesis and the last inequality is due to \cref{eq:etaLowerBound}.     
\end{proof}

\end{proof}

    \ifauthors
    \section{Acknowledgement}
    The second author thanks Or Zamir for bringing \cite{AL2013} to his attention.
    \fi

\printbibliography
\end{document}